\documentclass[10pt,a4paper]{llncs}
\usepackage{graphicx}
\usepackage[boxed,vlined,linesnumbered]{algorithm2e}
\usepackage{amsfonts}
\usepackage{amsmath}
\usepackage[usenames,dvipsnames,svgnames,table]{xcolor}
\usepackage{diagbox}
\usepackage{relsize}
\usepackage{framed}
\usepackage{url}
\usepackage[table]{xcolor}
\usepackage{tikz}
\usepackage{multirow}
\usepackage{times}
\usepackage{wrapfig}
\usepackage{subfig}
\usetikzlibrary{arrows,positioning,automata,shadows,shapes}

\usepackage[a4paper]{geometry}
\usepackage{url}
\usepackage[]{todo}
\usepackage{setspace}
\usepackage[belowskip=-14pt,aboveskip=0pt]{caption}
\let\subparagraph\paragraph
\usepackage[compact]{titlesec}

\SetKwProg{Fn}{function}{\string()}{}
\SetKwProg{When}{when}{\string:}{}
\SetKwProg{Repeat}{repeat until}{\string:}{}
\SetKwProg{Wait}{wait until}{\string :}{}
\DontPrintSemicolon
\newcommand{\secref}[1]{Sec.~\ref{#1}}
\newcommand{\figref}[1]{Fig.~\ref{#1}}
\newcommand{\true}{\mathit{true}}
\newcommand{\false}{\mathit{false}}

\newcommand{\defas}{\ensuremath{\stackrel{\rm\scriptscriptstyle def}{=}}}

\newenvironment{algo}[2]%
{\par\vspace{1em}\noindent\textbf{Algorithm~#1}~(\emph{#2}).~~}
{}

\newcommand{\toolnameshort}{\textsc{DecentMon2}}

\def\N{\mathbb{N}}
\def\B{\mathbb{B}}

\newcommand{\cQ}{{\cal Q}}

\DeclareMathOperator{\dom}{dom}

\newcommand{\ltlG}{\mathbf{G}}
\newcommand{\ltlF}{\mathbf{F}}
\newcommand{\ltlU}{\mathbf{U}}
\newcommand{\ltlX}{\mathbf{X}}
\newcommand{\ltlW}{\mathbf{W}}
\newcommand{\ltlR}{\mathbf{R}}
\newcommand{\ltlP}{\overline{\mathbf{X}}}

\newcommand{\qinit}{q_{\tiny\rm init}}
\newcommand{\AP}{\ensuremath{\mathit{AP}}}
\DeclareMathOperator{\good}{good}
\DeclareMathOperator{\bad}{bad}
\DeclareMathOperator{\verdict}{verdict}
\DeclareMathOperator{\mem}{mem}

\DeclareMathOperator{\leadermon}{leader\_mon}
\DeclareMathOperator{\choosemon}{choose\_mon}
\graphicspath{{fig/}}
\begin{document}
\title{Efficient and Generalized Decentralized Monitoring\\ of Regular Languages
}
\author{
Tom Cornebize \and Yli\`es Falcone\\
\email{Tom.Cornebize@gmail.com, Ylies.Falcone@ujf-grenoble.fr}}
\institute{Universit\'e Joseph Fourier Grenoble I, Laboratoire d'Informatique de Grenoble, Grenoble, France}
\maketitle
\begin{abstract}
The main contribution of this paper is an efficient and generalized decentralized monitoring algorithm allowing to detect satisfaction or violation of any regular specification by local monitors alone in a system without central observation point.
Our algorithm does not assume any form of synchronization between system events and communication of monitors, uses state machines as underlying mechanism for efficiency, and tries to keep the number and size of messages exchanged between monitors to a minimum.
We provide a full implementation of the algorithm with an open-source benchmark to evaluate its efficiency in terms of number, size of exchanged messages, and delay induced by communication between monitors.
Experimental results demonstrate the effectiveness of our algorithm which outperforms the previous most general one along several (new) monitoring metrics.
\end{abstract}
\section{Introduction}
\label{sec:intro}
%
Monitoring is a verification technique based on runtime information.
From a practical perspective, a decision procedure, the so-called \emph{monitor}, analyzes a sequence of events (or a trace) from the system under scrutiny, and emits verdicts w.r.t. satisfaction or violation of a specification formalized by a property.
Being lightweight is an important feature of monitoring frameworks because the performance of the system should be disturbed in a minimal way.
When the monitor collects events from a monolithic system, we refer to this as \emph{centralized monitoring}.

Modern systems are in essence distributed: they consist of several computation units (referred to as components in the sequel), possibly interacting together, and evolving independently.
Monitoring distributed systems is a long-standing problem.
The main challenge is to design algorithms that allow to i) efficiently monitor computation units of a system, ii) let local monitors recompute a global state of the system with minimal communication, and iii) monitor against rich specifications.
Existing monitoring frameworks usually assume the existence of a central observation point in the system to which components have to send events to determine verdicts; as seen for instance in~\cite{FalconeJNBB11,ZhouSLL09}.
In that case, from a theoretical perspective, monitoring reduces to the centralized case.
A more challenging situation occurs when such central observation point cannot be introduced in the system (because of architectural reasons or because communication should be minimized).
We refer to this as \emph{decentralized monitoring}.
In the decentralized setting, monitors emit verdicts with incomplete information: local monitors read local traces, i.e., incomplete versions of the global trace, and have to communicate with each other to build up a global verdict.
\paragraph{Related Work.}
Several approaches exist for monitoring distributed systems.
A temporal logic, M\textsc{t}TL, for expressing properties of
asynchronous multi-threaded systems was presented in~\cite{SenVAR06a}.
Its monitoring procedure takes as input a
\emph{safety} formula 
and a partially-ordered execution of a parallel asynchronous system.
M\textsc{t}TL augments linear temporal logic (LTL)~\cite{TemporalLogicOfProgram} with modalities related to the distributed/multi-threaded nature of the system under scrutiny.
Several works like \cite{DBLP:conf/fm/GenonMM06} target physically
distributed systems and address the monitoring problem of partially-ordered traces, and
introduce abstractions to deal with the combinatorial explosion of
these traces.

Close to our work is an approach to monitoring violations of invariants in distributed systems using knowledge~\cite{GrafPQ11}.
Model-checking the system allows to pre-calculate the states where a violation can be reported by a process alone.
When communication (i.e., more knowledge) is needed between processes, synchronizations are added.
Both~\cite{GrafPQ11} and our approach try to minimize the communication induced by the distributed nature of the system.
The main differences between our approach and~\cite{GrafPQ11} are that~\cite{GrafPQ11} requires the property to be stable (and considers only invariants) and uses a Petri net model to compute synchronization points.
We do not assume any model of the system, i.e., we consider the system as a black box.

Decentralized monitoring is also somewhat related to diagnosis of discrete-event systems which has the objective of detecting the occurrence of a fault after a finite number of steps, see for instance~\cite{lafortuneetal:2004,Cassez10}.
There are two main differences between monitoring and diagnosis.
In diagnosis, a specification with normal and faulty behavior is an input to the problem, whereas we consider the monitored system as a black box.
Also, when considering observability of distributed systems, diagnosis assumes a central observation point which may not have full access to information.
On the contrary, decentralized monitoring does not assume a central observation point, but that local monitors have access to all local information.
Similarly, decentralized observation~\cite{tripakis:2005:cdc-ecc} uses a central observation point in a system that collects verdicts from local observers that have limited memory to store local traces.
Note also that, neither diagnosis nor observability considers the problem of minimizing the communication overhead.

In previous work~\cite{BauerF12}, we proposed a decentralized monitoring algorithm for (all) LTL formulas.
The main novelties were to i) avoid the need for a central observation point in the system and ii) try to reduce the communication induced by monitoring by minimizing the number of messages exchanged between monitors.
The approach in~\cite{BauerF12} uses LTL specifications ``off-the-shelf'' by allowing the user to abstract away from the system architecture and conceive the system as monolithic.
The algorithm relied on a decentralized version of \emph{progression}~\cite{Bacchus98}: at any time, each monitor carries a temporarily extended goal which represents the formula to be satisfied according to the monitor that carries it.
The monitor rewrites its goal according to local observations and goals received from other monitors.
According to the propositions referred in the obtained formula, it might communicate its local obligation to other monitors.
Our approach relied on the perfect synchrony hypothesis (i.e., neither computation nor communication takes time) where communication relied on a synchronous bus.
This hypothesis is reasonable for certain critical embedded systems e.g., in the automotive domain (cf.~\cite{BauerF12} for more arguments along this line).
Moreover, it has been recently shown that this approach does not only ``work on paper" but can be implemented when finding a suitable sampling time such that the perfect synchrony hypothesis holds~\cite{Bartocci13}.

Nevertheless, to facilitate the application of~\cite{BauerF12} in more real scenarios, several directions of improvement can be considered.
First, it is assumed in~\cite{BauerF12} i) that at each time instant, monitors receive an event from the system and can communicate with each others, and ii) that communication does not take time.
Second, the approach used LTL formulas to represent the local state of the monitor and progression (i.e., formula rewriting) each time a new event is received.
A downside of progression, is the continuous growth of the size of local obligations with the length of trace; thus imposing a heavy overhead after 100 events.
Finally, while~\cite{BauerF12} minimizes communication in terms of number of messages (i.e., obligations), it neglects their (continuously growing) size, with the risk of oversizing the communication device, in practice.
\paragraph{Originality.}
In this paper, we propose to overcome the aforementioned drawbacks of~\cite{BauerF12} and make important generalization steps for its applicability.
First, instead of considering input specifications as LTL formulas we consider (deterministic) finite-state automata and can thus handle all regular languages instead of only counter-free ones.
Thanks to the finite-word semantics of automata, we avoid the monitorability issues induced by the infinite-word semantics of LTL~\cite{BauerLS06,FalconeFM12}.
Interestingly, using an automata-based structure is more runtime efficient than rewriting.
Second, in practice, communication and reception of events might not occur at the same rate or the communication device might become unavailable during monitoring.
Our algorithm allows desynchronization between the reception of events from the system and communication between monitors but also arbitrarily long periods of absence of communication, provided that a global clock exists in the system.
Our algorithm is fully implemented in an open-source benchmark.
Our experimental results demonstrate that our algorithm i) leads to a more lightweight implementation, and ii) outperforms the one in~\cite{BauerF12} along several (new) monitoring metrics.
\paragraph{Overview of the decentralized monitoring algorithm.}
Let $\mathcal{C} = \{ C_1, \ldots, C_n \}$ be the set of system components.
Let $L$ be a regular language formalizing a requirement over the system global behavior, i.e., the global requirement does not take into account the system structure.
Let $\tau_i = \tau_i(0) \cdots \tau_i(t)$ be the local behavioral trace on component $C_i$ at time $t\in \mathbb{N}$.
Further, let $\tau = \tau_1(0) \cup \ldots \cup \tau_n(0) \cdot \tau_1(1) \cup \ldots \cup \tau_n(1) \cdots \tau_1(t) \cup \ldots
\cup \tau_n(t)$ be the global behavioral trace, at time $t \in \mathbb{N}$, obtained by merging local traces.
An hypothesis of our framework is thus the existence of a global clock in the system.
From $L$, one can construct a \emph{centralized monitor} for $L$, i.e., a decision procedure having access to the global trace $\tau$ and emitting verdict $\top$ (resp. $\bot$) whenever $\tau$ is a good (resp. bad) prefix for $L$, i.e., whenever $\tau\cdot\Sigma^*\subseteq L$ (resp. $\tau\cdot \Sigma^* \subseteq (\Sigma^*\setminus L)$.
Then, from a centralized monitor, we define its \emph{decentralized version}, i.e., a monitor keeping track of possible evaluations of a centralized monitor when dealing with partial information about the global trace.
A copy of the decentralized monitor is attached to each component.
The decentralized monitor $\mathit{DM}_i$ attached to component $C_i$ reads the local trace $\tau_i = \tau_i(0) \cdots \tau_i(t)$, in separation.
Our decentralized monitoring algorithm orchestrates communication between monitors and how they exchange information about their received events or their evaluation of the current global state.
Communication between monitors is performed through messages and is not synchronized with the production of events on the system.
In other words, when a monitor sends some message to another one, there is no special assumption about the arrival time.
Furthermore, we assume that no message is lost when monitors communicate with each other.

The decentralized monitoring algorithm evaluates the global trace $\tau$ by reading each local trace $\tau_i$ of $C_i$, in separation.
In particular, it exhibits the following properties.
\begin{itemize}
\item
If a local monitor yields the verdict $\bot$ (resp. $\top$) on some component $C_i$ by observing $\tau_i$, it implies that $\tau\cdot \Sigma^* \subseteq \Sigma^*\setminus L$ (resp. $\tau\cdot \Sigma^* \subseteq L$) holds.
That is, a locally observed violation (resp. satisfaction) is, in fact, a global violation (resp. satisfaction).
\item
If the monitored global trace $\tau$ is such that $\tau\cdot\Sigma^* \subseteq \Sigma^* \setminus L$ (resp. $\tau\cdot \Sigma^* \subseteq L$), at some time $t$, one of the local monitors on some component $C_i$ yields $\bot$ (resp. $\top$), at some time $t' \geq t$ because
of some latency induced by decentralized monitoring, whatever is the global trace between $t$ and $t'$.
\end{itemize}
\paragraph{Paper Organization.}
The rest of this paper is organized as follows.
Section~\ref{sec:prelim} introduces some preliminaries and notations.
Section~\ref{sec:cm} proposes a generic (centralized) monitoring framework, compatible with frameworks that synthesize monitors in the form of finite-state machines.
Section~\ref{sec:dm} shows how to decentralize a monitor.
In \secref{sec:com}, we present how decentralized monitors communicate with each other to obtain a verdict in a decentralized manner.
Section~\ref{sec:sem} describes the relation between centralized and decentralized monitoring.
Section~\ref{sec:expe} presents our benchmark, {\toolnameshort}, used to evaluate an implementation of our monitoring algorithm.
Section~\ref{sec:fw} presents some perspectives.
\section{Preliminaries and Notations}
\label{sec:prelim}
%
$\N$ is the set of natural numbers.
For $i,j \in \N$, the (underlying set associated to the) interval of integers from $i$ to $j$ is denoted by $[i;j]$.
Given a finite set $E$, the set of finite sequences over $E$ is noted $E^*$.

We consider that the global system consists of a set of components $\{C_1,\ldots,C_n\}$, with $n \in \N\setminus\{0\}$.
Each component emits events synchronously and has a local monitor attached to it.
An event local to component $C_i$ is built over a set of atomic propositions $\AP_i$, $i\in[1;n]$, i.e., the local set of events is $\Sigma_i = 2^{\AP_i}$.
The set of all atomic propositions is $\AP = \cup_{i\in [1;n]} \AP_i$.
Atomic propositions are local to components by requiring that $\{AP_i \mid i \in [1;n]\}$ is a partition of $\AP$.
(Note, this hypothesis simplifies the presentation of the results in the paper but is not an actual limitation of our framework.)
The set of all local events in the system is $\cup_{i\in [1;n]} \Sigma_i$, where $\Sigma_i$ is visible to the monitor at component $C_i$, $i\in[1;n]$.
The global specification refers to events in $\Sigma = 2^{\AP}$ and is given by a regular language $L\subseteq \Sigma^*$.
Note that the specification does not take into account the architecture of the system and may refer to events involving atomic propositions from several components (i.e., $\Sigma \neq \cup_{i\in [1;n]} \Sigma_i$ in the decentralized case whereas $\Sigma = \cup_{i\in [1;n]} \Sigma_i$ in the centralized one or when there is only one component).
We assume that the (regular) language to be monitored is recognized by a deterministic finite-state automaton $(Q, \qinit, \Sigma, \delta, F)$ where $Q$ is the set of states, $\qinit\in Q$ the initial state, $\delta$ the transition function, and $F\subseteq Q$ the set of accepting states.

Over time, for $i \in [1;n]$, the monitor attached to component $C_i$ receives a trace $\tau_i\in (2^{\AP_i})^\ast$, a finite sequence of local events, representing the behavior of $C_i$.
The global behavior of the system is given by a global trace $\tau = (\tau_1 , \tau_2 , \ldots , \tau_n )$.
The global trace is a sequence of pair-wise union of the local events in component’s traces, each of which at time $t$ is of length $t + 1$ i.e., $\tau = \tau(0) \cdots \tau(t)$, where for $i\ < t$, $\tau(i)$ is the (i+1)-th element of $\tau$.
The sub-sequence $\tau[i; j]$ is the sequence containing the (i+1)-th to the (j+1)-th elements.
The substitution of the element at index $t$ in a sequence $\tau$ by the element $e$ is noted $\tau[t | e]$.
\section{Centralized Monitoring of (Propositional) Regular Languages}
\label{sec:cm}
%
In this section we propose a general framework for centralized monitoring of regular languages.
This framework is general enough to be compatible with most of the existing monitoring frameworks dedicated to propositional regular languages.

In the centralized case, the monitor is a central observation point.
Generally speaking, the purpose of the monitor is to determine whether the observed sequence forms a good or a bad prefix of the language being monitored.
For this purpose, the monitor emits verdicts in some truth-domain $\B$ s.t. $\{\bot,\top\} \subset \B$ where $\top$ and $\bot$ are two ``definitive values" used respectively when a validation (good prefix) and violation (bad prefix) of the language has been found, respectively.
\begin{definition}[Good and bad prefixes~\cite{DBLP:journals/tosem/BauerLS11}]
The sets of good and bad prefixes of a language $L\subseteq \Sigma^\ast$ are defined as:
\[
\good(L) = \{\tau\in\Sigma^\ast\mid \tau\cdot\Sigma^\ast \subseteq L\}, \qquad \bad(L) = \{\tau\in\Sigma^\ast\mid \tau\cdot\Sigma^\ast \subseteq (\Sigma^\ast\setminus L)\}.
\vspace{-0.5em}
\]
\end{definition}
Using good and bad prefixes, we can define the centralized semantic relation $\models_C$ for traces, using, for instance, the truth-domain $\B \defas\{\bot,?,\top\}$, where the truth-value $?$ indicates that no verdict has been found yet.
Given $\tau\in\Sigma^*$, we say that $\tau \models_C L = \top$ (resp. $\bot$) whenever $\tau\in\good(L)$ (resp. $\bad(L)$) and $\tau \models_C L = ?$ otherwise.
\begin{definition}[Centralized Monitor]
\label{def:centralized_mon}
A centralized monitor is a tuple $(Q,\Sigma,q_0,\delta,$ $\verdict)$ where $Q$ is the set of states, $\Sigma = 2^{\AP}$ the alphabet of events, $q_0
$ the initial state, $\delta : Q \times \Sigma \rightarrow Q$ the complete transition function, and $\verdict: Q \rightarrow \B$ is a function that associates a truth-value to each state.
\end{definition}
\begin{wrapfigure}{r}{0.29\textwidth}
\vspace{-10pt}
\centering
\includegraphics[scale=1]{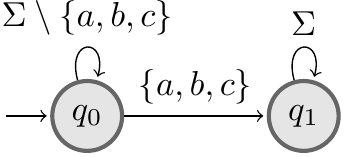}
\caption{{Transitions of $\mathit{CM}_1$}}
\label{fig:cm1}
\end{wrapfigure}

\noindent
A monitor is a Moore automaton, processing events from its alphabet, and emitting a verdict upon receiving each event.
Monitor-synthesis algorithms ensure that i) for any $\tau\in\Sigma^\ast$, $\verdict(\delta(q_0, \tau))=\top/\bot$ iff $\tau\in\good/\bad(L)$, where $\delta$ is extended to sequences in the natural way; ii) for any  $q \in Q$, if $\verdict(q) \in \{\top,\bot\}$ then $\forall \sigma \in \Sigma : \delta(q,\sigma) = q$.
A centralized monitor is thus a decision procedure w.r.t. the centralized semantics relation $\models_C$.

\begin{remark}[Truth-domains]
More involved truth-domains with refined truth-values (e.g., the ones used in~\cite{DBLP:journals/tosem/BauerLS11,FalconeFM12}) can be used in our framework without any particular difficulty.
\end{remark}
\begin{example}[Centralized Monitor]
\label{ex:cm1}
Consider $\AP^1 = \{a,b,c\}$ and $L_1$ the language of words over $2^{\AP^1}$ that contain at least one occurrence of the event $\{a,b,c\}$.
The monitor $\mathit{CM}_1$ of this language has its transition function $\delta_1$ depicted in~\figref{fig:cm1}.
Moreover, $\verdict(q_0) = ?$ and $\verdict(q_1) = \top$.
Consider $\tau_1 = \emptyset\cdot\{a,b\}\cdot\{a,b,c\}\cdot\{a\}$, we have $\emptyset\cdot\{a,b\}\cdot\{a,b,c\}\in\good(L_1)$ and $\tau_1\in\good(L_1)$.
\end{example}
\section{Decentralizing a Monitor}
\label{sec:dm}
%
Let us now use the previous example to see what would happen when using a centralized monitor on a local component where only a subset of $\AP$ can be observed.
Let us consider a simple architecture with three components $C_A, C_B, C_C$ respectively with sets of atomic propositions $\AP^1_A=\{a\}, \AP^1_B=\{b\}, \AP^1_C=\{c\}$.
If we use a central monitor on, say $C_A$, no event (in $2^{\AP^1_A}$) could allow the monitor to reach $q_1$.
Monitors should thus take into account what could \emph{possibly} happen on other components.
Given an observation on a local component, a decentralized monitor computes the set of states that are \emph{possible} with this observation, and refines (i.e., eliminate possible states) when communicating with other monitors (as we shall see in \secref{sec:com}).

Given a centralized monitor, we define its decentralized version as follows.
\begin{definition}[Decentralized Monitor]
Given a centralized monitor $(Q,\Sigma,q_0,\delta,$ $\verdict)$, the associated decentralized monitor is a 5-tuple $(2^Q\setminus\{\emptyset\},(2^{[1;n]} \setminus \{\emptyset\}) \times \Sigma,\{q_0\},\Delta_\delta,$ $\verdict_D)$ where:
\begin{itemize}
\item
$(2^{[1;n]} \setminus \{\emptyset\}) \times \Sigma$ is the alphabet,
\item
$\Delta_\delta: (2^Q\setminus \{\emptyset\}) \times (2^{[1;n]} \setminus \{\emptyset\}) \times \Sigma \rightarrow (2^Q\setminus \{\emptyset\})$ is the decentralized transition function defined as:\\
$
\Delta_\delta(\cQ, s, \sigma) = \{q'\in Q \mid \exists \sigma' \in \Sigma, \exists q \in \cQ: \sigma = \sigma' \cap \bigcup_{j\in s} \AP_j \wedge q' = \delta(q,\sigma')\},
$
\item
$\verdict_D: (2^Q\setminus \{\emptyset\}) \rightarrow \B$ is the decentralized verdict function, s.t.:
\[
\verdict_D (\cQ) =
\left\{
\begin{array}{ll}
b & \text{if } \exists b \in \B: \{\verdict(q) \mid q \in \cQ\} = \{b\},\\
? & \text{otherwise},
\end{array}
\right.
\]
for any $\cQ\in 2^Q\setminus \{\emptyset\}$.
\end{itemize}
\label{def:decentralized_mon}
\end{definition}
Intuitively, a decentralized monitor ``estimates" the global state that would be obtained by a centralized monitor observing the events produced on all components.
The estimation of the global state is modeled by a set of possible states (of the centralized monitor) given the (local) information received so far.
When a decentralized monitor receives an event $(s,\sigma)$, it is informed that the union of the atomic propositions that occurred on the components indexed in the set $s$ is $\sigma$.
The transition function is s.t. if the estimated global state is $\cQ\in 2^Q\setminus \{\emptyset\}$ and it receives $(s,\sigma)$ as event, then the estimated global state changes to $\Delta_\delta(\cQ, s, \sigma)$ which contains all states s.t. one can find a transition in $\delta$ from a state in $\cQ$ labeled with a global event $\sigma'$ compatible with $\sigma$.
In other words, if the actual global state belongs to $\cQ$, and the union of events that happen on components indexed in $s$ is $\sigma$, then the actual global state belongs to $\Delta_\delta(\cQ, s, \sigma)$ which is the set of states that can be reached from a state in $\cQ$ with all possible global events (obtained by any observation that could happen on components indexed in $[1;n]\setminus s$).
Regarding verdicts, a decentralized monitor emits the same verdict as a centralized one when the current state contains states of the centralized monitor that evaluate on the same verdict.
\begin{figure}[t]
\centering
\includegraphics[scale=1]{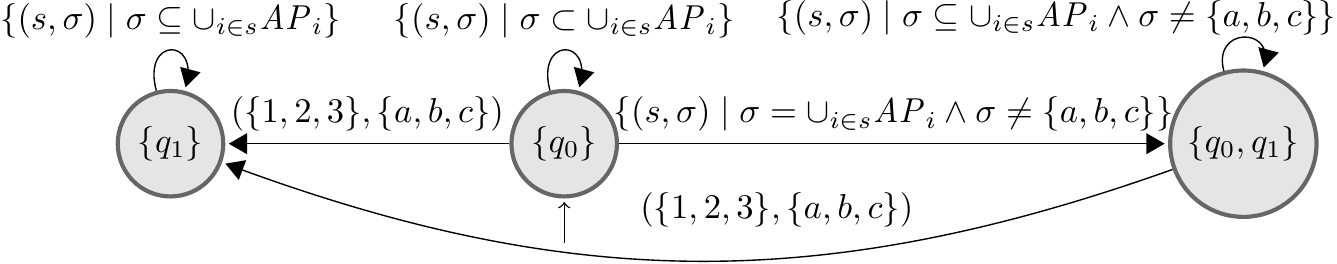}
\caption{Transitions of $\mathit{DM_1}$}
\label{fig:dm1}
\end{figure}
\begin{remark}[Verdict function]
The proposed verdict function allows a decentralized monitor to emit a verdict even if it does not ``know" the global state.
Moreover, alternative verdict functions (that would, for instance, return a set of verdicts from the centralized monitor) are possible.
\end{remark}
\begin{example}[Decentralized Monitor]
Let us consider again the architecture and language $L_1$ of Example~\ref{ex:cm1}.
Consider what happens initially on any of the components executing $\mathit{DM_1}$, the decentralized version of $\mathit{CM_1}$, see \figref{fig:dm1}.
Initially, the estimated global state is $\{q_0\}$.
Suppose the monitor is informed that $\{a\}$ occurred on component $C_A$ (of index 1), then it will change its estimated global state to $\Delta_{\delta_1}(\{q_0\}, \{1\}, \{a\}) = \{q_0, q_1\}$.
Intuitively, this transition can be understood as follows.
Knowing that $\{a\}$ occurred on $C_A$, the other possible global events are $\{a,b\}$ and $\{a,b,c\}$, as the monitor does not have information on what happened on $C_B$ and $C_C$.
In $\mathit{CM_1}$, from state $q_0$ and these events, states $q_0$ and $q_1$ can be reached.
Note, the only way to reach $\{q_1\}$ in $\mathit{DM_1}$, i.e., to know that the global state is $q_1$ (and is unique), $\mathit{DM_1}$ has to know that the union of events that occurred on components indexed in $\{1,2,3\}$ is $\{a,b,c\}$.
\end{example}
As illustrated by the previous example, a decentralized monitor does not depend on the component on which it executes.
Its transitions can occur on any component, as it receives an event together with the identifier of components on which such an event occurred.
However, a decentralized monitor is not purposed to be used alone but shall communicate with other decentralized monitors.
\section{Communication and Decision Making}
\label{sec:com}
\newcommand{\fctreceive}{
\begin{minipage}{0.54\linewidth}
\begin{algorithm}[H]
$(\mathit{rcv\_mem},\mathit{rcv\_state}) \longleftarrow (\mathit{false},\mathit{false})$ \;
	\When{an event $\sigma \in \Sigma_i$ is received from component}{
		$t \longleftarrow t + 1$ \;
		$\mem \longleftarrow \mem \sqcup\, [t \mapsto (\sigma,\{i\})]$\;
	}
	\When{a state $q' \in Q$ is received with time $t_{\rm new}$}{
		\If{$t_{\rm new} > t_{\rm last}$}{
			$(q,t_{\rm last}) \longleftarrow (q',t_{\rm new})$ \;
			$\mathit{rcv\_state} \longleftarrow \mathit{true}$ \;		
		}
	}
	\When{a partial memory $m \in \N \rightarrow$ $ \Sigma \times (2^{[1;n]}\setminus\{\emptyset\})$ is received}{
		$\mem \longleftarrow \mem \sqcup\, m$ \;
		$\mathit{rcv\_mem} \longleftarrow \mathit{true}$ \;		
	}
\caption{function $\mathsf{receive}$}
\label{algo:receive}
\end{algorithm}
\end{minipage}
}
\newcommand{\fctupdstate}{
\begin{minipage}[t]{0.45\linewidth}
\begin{algorithm}[H]
	$\cQ \longleftarrow \{q\}$ \;
	$\mathit{upd\_state}\longleftarrow \mathit{false}$ \;
	\For{$t'$ \textbf{from} $t_{\rm last}$ \textbf{to} $t$}{
		$(\sigma,s) \longleftarrow \mem(t')$ \;
		$\cQ \longleftarrow \Delta_\delta(\cQ,s,\sigma)$   \;
		\If{$\exists q' \in Q: \cQ = \{q'\}$}{
			$(q,t_{\rm last}) \longleftarrow (q', t' + 1)$\;
			$\mathit{upd\_state}\longleftarrow \mathit{true}$ \;
		}		
}
\caption{function $\mathsf{update\_state}$}
\label{algo:updstate}
\end{algorithm}
\end{minipage}
}
\newcommand{\thealgo}{
\begin{algorithm}[H]
$(t_{\rm last},t) \longleftarrow (0,-1) $ \;
$(q, \cQ) \longleftarrow (q_0, \{q_0\}) $ \;
$\mem \longleftarrow \{\,\}$ \;
\Repeat{the end of the trace and $t_{\rm last} > t$}{
	initialize \textit{message} \;
	$\mathsf{receive}()$ \;
	$\mathsf{update\_state}()$\;
	\If{$\mathit{upd\_state}\vee \mathit{rcv\_state}$}{
		\If{$\verdict(q) \in \{\top,\bot\}$}{
			\Return $\verdict(q)$\;
		}
		add $(q, t_{\rm last})$ to \textit{message} \;		
	}
	\If{$t_{\rm last} \leq t \wedge \left(\mathit{rcv\_mem}\vee \leadermon(i)\right)$}{
		add $\left(\mem\left(t_{\rm last},t\right), t_{\rm last}\right)$ to \textit{message} \;
	}
	\If{$\mathit{message}$ is not empty}{
		send $\mathit{message}$ to $M_{\choosemon(i)}$ \;
	}
}
\Return {$\verdict(q)$}
\caption{Decentralized monitoring algorithm executing on $C_i$ (main loop)}
\label{algo:dmon}
\end{algorithm}
}
Our aim is now to define how a collection of decentralized monitors, analyzing a given distributed trace, should communicate with each other to obtain a verdict in a decentralized manner.
The verdict indicates whether the trace, when interpreted as a global trace, is a good or a bad prefix of the language.
%
\subsection{Preliminaries: Local Memory, Clocks, and Communication}
\label{sec:com:prelim}
%
%
\paragraph{Monitor local memory.}
The local memory of a monitor is a partial function $\mem: \N \rightarrow \Sigma \times (2^{[1;n]}\setminus \{\emptyset\})$, purposed to record the ``local knowledge" w.r.t. (past instants of) the global (actual) trace produced by the system.
If $\mem(t) =  (\sigma_t,s_t)$, it means that the monitor knows that the set of all atomic propositions received by the components in $s_t$ is $\sigma_t$.
Moreover, if $\sigma \in \Sigma$ is the global event at time $t$ and $\mem(t) = (\sigma_t,s_t)$, then $\sigma \cap (\bigcup_{i\in s} \AP_i) = \sigma_t$.
In next section, we will see how after communicating, local monitors can discard elements from their memory.

As a local monitor memorizes the observed local events, it may inform other monitors of the content of its memory via messages.
When a monitor receives a memory chunk from another monitor, it merges it with its local memory.
For this purpose, for two memories $\mem$ and $\mem'$, we define the merged memory $\mem \sqcup \mem'$, as follows:
\[
(\mem \sqcup \mem') (t)
=
\left\{
\begin{array}{ll}
\mem(t) \cup \mem'(t) & \text{if } t\in\dom(\mem) \cap \dom(\mem'),\\
\mem'(t) & \text{if } t\in \overline{\dom(\mem)} \cap \dom(\mem'),\\
\mem(t) & \text{otherwise},
\end{array}
\right.
\]
where the union $(\sigma,s) \cup (\sigma',s')$ between two memory elements $(\sigma,s)$ and $(\sigma',s')$ is defined as $(\sigma \cup \sigma',s \cup s')$.
For instance, consider $\mem = \{ 0 \mapsto (\{b\},\{1,2\}),1\mapsto (\{a,b\},\{1,2\}), 2\mapsto (\emptyset,\{2\}) \}$ and $\mem' = \{ 1\mapsto (\{c\},\{3\}), 2\mapsto (\{c\},\{3\})\}$, we have $\mem \sqcup \mem' = \{ 0 \mapsto (\{b\},\{1,2\}), 1 \mapsto (\{a,b,c\},\{1,2,3\}), 2\mapsto (\{c\},\{2,3\}) \}$.
%
%
\paragraph{Monitor local clocks.}
Each local monitor carries two local (discrete) clocks $t$ and $t_{\rm last}$.
The purpose of $t$ is simply to store the time instant of the last received event from the local component.
The purpose of $t_{\rm last}$ is to store the time instant for which it knows the global state of the system.
Indeed, the decentralized monitoring algorithm presented in next section will ensure that, on each monitor $M_i$, for a global trace $\tau$:
\begin{itemize}
\item the last event $\sigma$ emitted by the local component was at time $t: \sigma = \tau(t)$.
\item the current state is the state corresponding to $t_{\rm last}$ : $q = \delta(q_0,\tau[0;t_{\rm last}-1])$; 
\end{itemize}
\paragraph{How monitors communicate.}
As mentioned before, local monitors are required to communicate with each other to share collected information (from their local observation or other monitors).
To ensure that communication between monitors aggregates correctly information over time, we suppose having two functions $\leadermon$ and $\choosemon$ that can be defined e.g., according to the architecture and possibly changing over time.

The function $\choosemon : [1;n] \rightarrow [1;n]$ indicates for each monitor, the monitor it should communicate with.
Local monitors are referred to by their indexes.
For information to aggregate correctly, we require $\choosemon$ to be bijective, and such that $\forall i \in [1;n], \forall k \in [1;n-1] : \choosemon^k(i) \neq i$ where $\choosemon^k(i) = \underbrace{\choosemon(\dots(\choosemon(i) )\dots)}_{k\;\text{times}}$.
One can consider for instance $\choosemon(i) = (i \mod n)+1$.
Note: these requirements are not limitations of our framework but rather guidelines for configuring the communication of our monitors where the architecture is such that a bidirectional direct communication exists between any two components.
The proposed algorithms can be easily adapted to any other architecture, provided that a bidirectional communication path exists between any two components (which otherwise would limit the interest of decentralized monitoring).

The function $\leadermon: [1;n] \rightarrow \{\true,\false\}$ indicates whether the monitor on the component of the given index is a leader.
When receiving new events from the system, only leader monitors can send their observation.
The number of leader monitors influences communication metrics of the monitoring algorithm (see~\secref{sec:expe}).
Using a function makes the algorithm generic and allows leader monitors to change over time.
%
\subsection{Decentralized Monitoring Algorithm}
\label{sec:com:algo}
%
Let us now present the main algorithm for decentralized monitoring.
The algorithm is executed independently on each component until there is no event to read and the local monitor has determined the global state, which is given by the condition $t_{\rm last} > t$ (the time instant corresponding to the last known global state is greater than the time instant of the last received event from the local component).

At an abstract level, the algorithm is an execution engine using a decentralized monitor as per Definition~\ref{def:decentralized_mon}.
It computes the locally estimated global state of the system by aggregating information from events read locally and partial traces received from other monitors.
It stores in $q$ the last known global state of the system at time $t_{\rm last}$, and in $t$ the time instant of the last event received from the system.
The main steps of the algorithm can be summarized as follows:
\begin{algo}{DM}{Decentralized Monitoring}
Let $L$ be the language to be monitored and $q_0$ the initial state of its associated centralized monitor.
Initialize variables $q$ to $q_0$, $t_{\rm last}$ to $0$, and $t$ to $-1$.
\begin{description}
\item
[DM1.]
[Wait] for something from the outside: either an event $\sigma$ from the system or a message from another monitor (a pair $(q',t_{\rm new}) \in Q \times \N$ or a partial memory $m$).
\item
[DM2.]
[Update]
If an event (resp. a trace) is received from a component (resp. another monitor), update memory and $t$.
If a state is received, update the last known global state.
\item
[DM3.]
[Compute new state]
Using the definition of the transition function of the decentralized monitor (Definition~\ref{def:decentralized_mon}) and the local memory between $t_{\rm last}$ and $t$, compute the set of possible states.
If the set of possible states is a singleton, $q$ and $t_{\rm last}$ are updated.
\item
[DM4.]
[Evaluate and return]
If a definitive verdict ($\top$ or $\bot$) is found, return it (and inform other monitors).
\item
[DM5.]
[Prepare communication]
Prepare a message to be sent.
If a state is received or a new state has been computed (i.e., if $q$ and $t_{\rm last}$ have been modified), append it to the message together with $t_{\rm last}$.
If there are events that occurred after the last found state ($t\geq t_{\rm last}$), append them to the message, provided that the monitor is a leader ($\leadermon(i)=\mathit{true}$) or these events come from another monitor.
\item
[DM6.]
[Communicate]
If there is a non-empty message to be sent, then send it to the associated monitor (as determined by function $\choosemon(i)$).
Go back to step DM1.
\end{description}
\end{algo}
\fctreceive
\fctupdstate
\thealgo
At a concrete level, the abstract algorithm is realized in Algorithms~\ref{algo:receive}, \ref{algo:updstate}, and~\ref{algo:dmon}.
These algorithms execute in the same memory space, and variables are global. 
The $\mathsf{receive}$ function (Algorithm~\ref{algo:receive}) realizes steps \textbf{D1} and \textbf{D2} where i) events and messages from other monitors are received, and, ii) the memory and current state are updated.
The $\mathsf{receive}$ function is called by the main loop (Algorithm~\ref{algo:dmon}) and blocks the execution until an input is received.
It can receive three possible inputs (and any combination of them): an event $\sigma$ from the component (then it updates $\mem$ and $t$), a state $q'$ from another monitor (then it updates $q$ and $t_{\rm last}$ if it does not have fresher information), a partial memory $m$ from another monitor (then it updates $\mem$), or both a state and a partial memory.
The function also keeps track of whether a state or a partial memory was received using two Booleans $\mathit{rcv\_state}$ and $\mathit{rcv\_mem}$.
The $\mathsf{update\_state}$ function (Algorithm~\ref{algo:updstate}) realizes step \textbf{D3} by implementing the transition function $\Delta_\delta$ of the decentralized monitor using at the same time the local memory $\mem$ for efficiency reasons.
Variable $q$ keeps track of the last know global state (at time $t_{\rm last}$.
Variable $\cQ$ is a temporary variable that keeps track of the set of possible states.
Variable $\mathit{upd\_state}$ is set to $\mathit{true}$ if the execution of $\mathsf{update\_state}$ function allows to update the last know global state.
The main loop (Algorithm~\ref{algo:dmon}) realizes steps \textbf{D4}, \textbf{D5}, and \textbf{D6} where the message is built.
Step \textbf{D4} is realized by lines 8 to 11, where, if a new global state is known (either computed with $\mathsf{update\_state}$ or received in a message), then it is checked if the associated verdict is definitive.
The new state together with $t_{\rm last}$ are added to the message.
Then, when there are some local events to be shared ($t_{\rm last} \leq t$), if the monitor received a partial memory or the monitor is a leader (line 12), the partial memory from $t_{\rm last}$ to $t$ (i.e., $\mem(t_{\rm last},t)$) and the value of $t_{\rm last}$ are added to the message (line 13).
Finally (lines 14-15), the (non-empty) message is sent to the monitor of index $\choosemon(i)$.
\begin{example}[Decentralized Monitoring]
\label{ex:dm}
Let us go back to the monitoring of the specification introduced in Example~\ref{ex:cm1} and see how this specification is monitored with Algorithms~\ref{algo:receive}, \ref{algo:updstate}, and~\ref{algo:dmon}.
Table~\ref{tab:ex1} shows how the situation evolves on all three monitors when monitoring the global trace $\emptyset\cdot\{a,b\}\cdot\{a,b,c\}\cdot\{a\}$.
As mentioned earlier, the sequence of states of the centralized monitor is $q_0\cdot q_0\cdot q_1\cdot q_1$, and the verdict associated to this trace is $\top$, obtained after the third event.
For this example, $\leadermon(i) = (i=1)$ and $\choosemon(i) = (i \mod n) + 1$.
For simplicity, in this example, communication between monitors and events from the system occur at the same rate.
Cells are colored in grey when a communication occurs between monitors or an event is read from a component.
On each monitor, between any two communications or event receptions, the local memory is represented on two lines: first the values of $t_{\rm last}$, $t$, and $q$ the last determined global state, and second the memory content.
\end{example}
\begin{table}[t]
\centering
\caption{Decentralized monitoring of $L_1$ on 3 components}
\includegraphics[scale=0.98]{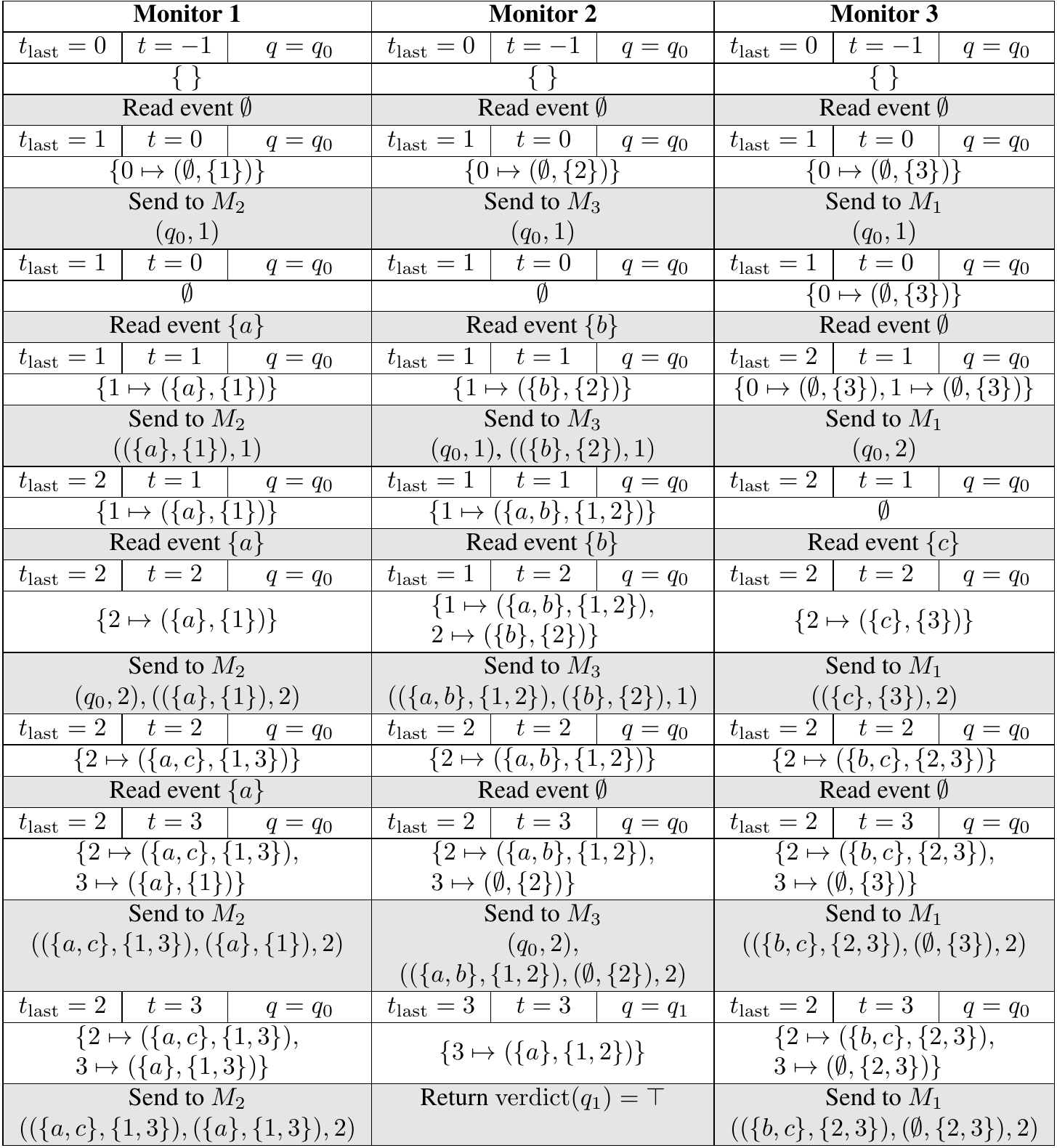}
\label{tab:ex1}
\end{table}
\begin{remark}[Domain of $\mem$]
At any moment, the only used elements of $\mem$ are those between $t_{\rm last}$ and $t$.
Thus, after each step of the algorithm, elements before $t_{\rm last}$ can be discarded.
Thus, $\dom(\mem) = [t_{\rm last};t]$ is of bounded size under certain conditions discussed in \secref{sec:sem}. 
\end{remark}
\begin{remark}[Optimizations]
Further optimizations can be taken into account in the algorithm.
For instance, using an history of sent messages, monitors can remove information from some messages addressed to another monitor, if they already sent this information in a previous message.
Further studies are needed to explore the trade-off between local memory consumption vs the size of exchanged messages in the system.
\end{remark}
\section{Semantics and Properties of Decentralized Monitoring}
\label{sec:sem}
%
In this section, we discuss further the semantics induced by the decentralized monitoring algorithm and its properties.
\begin{definition}[Semantics of Decentralized Monitoring]
Let $\mathcal{C} = \{ C_1, \ldots, C_n \}$ be the set of system
  components, $L\subseteq (2^\AP)^*$ be a regular language, and $\mathcal{M} = \{
  M_1, \ldots, M_n \}$ be the set of component monitors.
  Further, let $\tau = \tau_1(0) \cup 
\ldots \cup \tau_n(0) \cdot \tau_1(1) \cup 
\ldots \cup \tau_n(1) \cdots \tau_1(t) \cup \ldots
\cup \tau_n(t)$ be the global
  behavioral trace, at time $t \in \mathbb{N}$.
  If some component $C_i$, with $i \leq n$, $M_i$ has a local state $\cQ$ s.t. $\verdict_D(\cQ)  = \top$ (resp.\ $\bot$), then $\tau \models_D L = \top$
  (resp.\ $\bot$). Otherwise, 
  $\tau\models_D L = ?$.
\end{definition}
By $\models_D$ we denote the satisfaction relation on finite traces in
the decentralized setting to differentiate it from the centralized one.
Obviously, $\models_C$ and $\models_D$ both yield values from the same
truth-domain. However, the semantics are not equivalent, since the
current state of the decentralized monitor can contain several states of the centralized one, when a local component has not enough information to determine a verdict.
This feature was illustrated in Example~\ref{ex:dm} where at $t=2$, the global trace is $\emptyset\cdot \{a,b\} \cdot \{a,b,c\}$, which is a good prefix of the monitored language, only reported at $t=4$ by Monitor 2.

The precise relation between the centralized and decentralized semantics is given by the two following theorems.
\begin{theorem}[Soundness]
  Let $L\subseteq \Sigma^\ast$ and $\tau \in \Sigma^\ast$, then $\tau \models_D
  L = \top/\bot \Rightarrow \tau \models_C L = \top/\bot$, and
  $\tau \models_C L = {?}  \Rightarrow \tau \models_D L = {?}$.
\end{theorem}
Soundness states that i) all definitive verdicts found by the decentralized monitoring algorithm are actual verdicts that would be found by a centralized monitor, having access to the global trace, and ii) decentralized monitors do not find more definitive verdicts ($\top$ or $\bot$) than the centralized one.
\begin{theorem}[Completeness]
\label{theo:completeness}
Let $L\subseteq \Sigma^\ast$ and $\tau \in \Sigma^\ast$, then $\tau\models_C L= \top/\bot \Rightarrow \exists \tau'\in\Sigma^\ast: \tau\cdot \tau'\models_D L= \top/\bot$.
\end{theorem}
Completeness states that all verdicts found by the centralized algorithm for some global trace $\tau$ will be eventually found by the decentralized algorithm on a continuation $\tau\cdot\tau'$.
Generally, when the rate of communication between monitors (compared to the reception of events) is unknown or when not all monitors are leaders, it is not possible to determine the maximal length of $\tau'$.
When monitors communicate at the same rate as monitors receive events and all monitors are leaders (i.e., they can send message spontaneously -- $\leadermon(i) = \mathit{true}$, for any $i\in[1;n]$), then, as was the case in~\cite{BauerF12}, we can bound the maximal length of $\tau'$ by $n$ (the number of components in the system), which also represents the maximal delay, induced by decentralized monitoring.
\begin{theorem}[Completeness with bounded delay]
\label{theo:completeness2}
Let $L\subseteq \Sigma^\ast$ and $\tau \in \Sigma^\ast$, if monitors receive events and communicate at the same rate and if all monitors are leaders, then $\tau\models_C L= \top/\bot \Rightarrow \exists \tau'\in\Sigma^\ast: |\tau'| \leq n \wedge \tau\cdot \tau'\models_D L= \top/\bot$.
\end{theorem}
\section{Implementation and Experimental Results}
\label{sec:expe}
%
We present {\toolnameshort} a new benchmark tool used to evaluate decentralized monitoring (\secref{sec:expe:tool}) using specifications given as LTL formulas (\secref{sec:expe:ltl}) and specifications patterns (\secref{sec:expe:pat}).
Then, we draw conclusions from our experiments (\secref{sec:expe:conc}).
Further experimental results are available at~\cite{decenttool2}.
%
\subsection{{\toolnameshort}: a Benchmark for Generalized Decentralized Monitoring}
\label{sec:expe:tool}
%
\begin{figure}[t]
\centering
\includegraphics[scale=0.9]{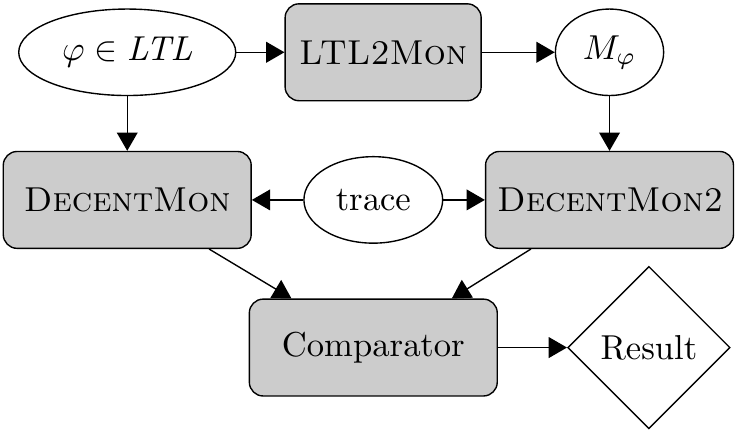}
\caption{Experimental setup for comparing \textsc{DecentMon} and {\toolnameshort}}
\label{fig:expe_setup}
\end{figure}
{\toolnameshort} is an benchmark dedicated to decentralized monitoring.
{\toolnameshort} that consists of:
\begin{itemize}
\item
a completely redeveloped version of \textsc{DecentMon}~\cite{BauerF12},
\item
an implementation of the decentralized monitoring algorithm presented in \secref{sec:com:algo},
\item
a trace generator, and
\item
an LTL-formula generator.
\end{itemize}
{\toolnameshort} consists of 1,300 LLOC, written in the functional programming language OCaml.
It can be freely downloaded and run from~\cite{decenttool2}.

The system takes as input multiple traces (that can be automatically generated), corresponding to the behavior of a distributed system, and a specification given by a deterministic finite-state automaton.
Then the specification is monitored against the traces in two different modes: a) by merging the traces to a single, global trace and then using a ``centralized monitor" for the specification (i.e., all components send their respective events to the central monitor who makes the decisions regarding the trace), b) by using the decentralized version introduced in~\cite{BauerF12}, and c) by using the decentralized approach introduced in this paper (i.e., each trace is read by a local monitor in the two last cases).
To favor the centralized case, monitors send their events only if they differ from the previous one, which decreases the number of exchanged messages.
We have evaluated the three different monitoring approaches (i.e., centralized vs. LTL-decentralized vs generalized-decentralized) using several set-ups described in the remainder of this section.
To compare monitoring metrics obtained with the decentralized algorithm in~\cite{BauerF12} and the one in this paper, we used \textsc{LTL2Mon}~\cite{ltl2mon}, to convert LTL formulas into automata-based (centralized) monitors.
For our comparison purposes, we used results on common LTL formulas and traces using the experimental setup depicted in \figref{fig:expe_setup}.
For each of the metric mentioned in the following sections, ratios are obtained by dividing the value obtained in the centralized case over the value obtained in the decentralized case.

To compare with the decentralized monitoring algorithm obtained in~\cite{BauerF12}, the emission of events occurs at the same rate as the communication between monitors.
Recall that it was assumed in~\cite{BauerF12} whereas our monitoring algorithm allows different ratios.

Each line of the following arrays is obtained by conducting 1,000 tests, each with a fresh trace of 1,000 events and specification.
We use the same architecture as in the running example.
Note that benchmarks with different architectures and rates of communication/event-emission were also conducted, and are available from~\cite{decenttool2}.

For the following monitoring metrics, we measure the size of the elements exchanged by monitors as follows.
Suppose we monitor an LTL formula $\varphi$ over $\AP$ with an automaton defined over the alphabet $\Sigma=2^\AP$ with set of states $Q$: each event is of size $\lceil \log_2 |\Sigma| \rceil$, each state is of size $\lceil \log_2 |Q| \rceil$, each time unit $t$ is of size $\lceil \log_2 (t) \rceil$, each formula is of size $n \times \lceil \log_2 (| \AP | + |\mathit{Op}|)\rceil$ where $n$ is the number of symbols in the formula, $\AP$ is the set of atomic propositions of the the formula and $\mathit{Op}=\{\top,\bot,\vee,\wedge,\neg,\Rightarrow,\Leftrightarrow,$ $\ltlX,\ltlF,\ltlG,\ltlU,\ltlR,\ltlW,\ltlP,\#,$ $(,)\}$ is the set of symbols in formulas handled by \textsc{DecentMon}.
Then in the following tables, the following metrics are used:
\begin{itemize}
\item
$\# \mathrm{msg.}$, the total number of exchanged messages,
\item
$|\mathrm{msg.}|$, the total size of exchanged messages (in bits),
\item
$|\mathrm{trace}|$ the size of the prefix of the trace needed to obtain a verdict,
\item
$\mathrm{delay}$, the number of additional events needed by the decentralized algorithm to reach a verdict compared to the centralized algorithm,
\item
$|\mathrm{mem}|$, the memory in bits needed for the structures (i.e., formulas for~\cite{BauerF12}, partial function $\mem$ plus state for our algorithm).
\end{itemize}
%
\subsection{Benchmarks for Randomly Generated formulas}
\label{sec:expe:ltl}
%
\begin{table}[t]
\centering
\caption{Experimental results for random formulas}
\setlength{\extrarowheight}{1mm}
\subfloat[Number and size of messages\label{tab:bench:random1}]{
\begin{tabular}{|c|r|r|r|r|r|r|r|r|r|r|r|}
\hline
$|\varphi|$ & \multicolumn{3}{c|}{$\# \mathrm{msg.}$} & \multicolumn{3}{c|}{$|\mathrm{msg.}|$} & \multicolumn{2}{c|}{$\# \mathrm{msg.}$ ratio} & \multicolumn{2}{c|}{$|\mathrm{msg.}|$ ratio} \\
\hline 
1 & 3.49 & 1.13 & 3.73 & 10.4 & 87.2 & 23.8 & 0.32 & 1.06 & 8.31 & 2.27\\
\rowcolor[gray]{0.9}
2  & 4.04 & 1.89 & 5.4 & 12.1 & 316 & 39.2 & 0.46 & 1.33 & 26.0 & 3.23\\
3 & 9.33 & 5.34 & 16.9 & 27.9 & 3,220 & 166 & 0.57 & 1.37 & 115 & 4.5\\
\rowcolor[gray]{0.9}
4 & 25.1 & 12.6 & 35.9 & 75.3 & 8,430 & 350 & 0.5 & 1.27 & 112 & 4.16\\
5 & 39.7 & 21.9 & 71.0 & 119 & 36,500 & 775 & 0.55 & 1.33 & 306 & 4.86\\
\rowcolor[gray]{0.9} 6 & 90.9 & 47.3 & 116 & 272 & 284,000 & 1,180 & 0.52 & 1.23 & 1,040 & 4.21\\
\hline
\end{tabular}
}

\subfloat[Trace length, delay, and memory size\label{tab:bench:random2}]{
\begin{tabular}{|c|r|r|r|r|r|r|r|r|r|}
\hline 
$|\varphi|$ & \multicolumn{3}{c|}{$|\mathrm{trace}|$} & \multicolumn{2}{c|}{$\mathrm{delay}$} & \multicolumn{2}{c|}{$|\mathrm{mem}|$} \\
\hline 
1 & 1.33 & 1.66 & 2.61 & 0.32 & 1.28 & 44.2 & 7.93\\
\rowcolor[gray]{0.9}
2 & 1.67 & 2.15 & 3.2 & 0.48 & 1.53 & 156 & 9.72\\
3 & 5.21 & 5.79 & 8.8 & 0.58 & 1.6 & 458 & 10.4 \\
\rowcolor[gray]{0.9}
4 & 15.7 & 16.4 & 19.3 & 0.7 & 1.66 & 1,100 & 11.3\\
5 & 25.5 & 26.4 & 36.3 & 0.82 & 1.79 & 2630 & 12.4\\
\rowcolor[gray]{0.9}
6 & 59.4 & 60.2 & 63.2 & 0.76 & 1.66 & 5,830 & 12.0\\
\hline 
\end{tabular}
}
\label{tab:bench:random}
\end{table}
For each size of formula (from 1 to 6), {\toolnameshort} randomly generated 1,000 formulas in the architecture described in Example~\ref{ex:cm1}.
How the three monitoring
 approaches compared on these formulas can be seen in
Tables~\ref{tab:bench:random1} and~\ref{tab:bench:random2}.
The first column of these tables shows the size of the monitored LTL formulas. 
Note, our system measures formula size in terms of operator
entailment\footnote{Experiments show that operator entailment is more representative of how difficult it is to
  progress it in a decentralized manner. formulas of size above $6$ are not realistic in practice.} inside it (state formulas excluded), e.g., $\ltlF a \wedge \ltlG (b\wedge c)$ is of size~$2$.

For example, the last line in Table~\ref{tab:bench:random1} says that we
monitored 1,000 randomly generated LTL formulas of size 6.
On average, monitors using the centralized algorithm, the decentralized algorithm using LTL formulas, and the decentralized algorithm using automata, exchanged 90.9, 47.3, 116 messages, had messages of size 272 bits, 284,000 bits, 1180 bits, respectively.
The last two pairs of columns show the ratios of the previous metrics obtained in the decentralized cases over the centralized one.
For instance, the last line in Table~\ref{tab:bench:random1} says that the decentralized algorithm with LTL formulas induced 0.52 times the number of messages of the centralized algorithm, whereas the decentralized algorithm with automata induced 1.23 times messages.
Message ratios and metrics in Table~\ref{tab:bench:random2} read similarly.
%
%
\subsection{Benchmarks for Patterns of formulas}
\label{sec:expe:pat}
%
We also conducted benchmarks with more realistic specifications, obtained from specification patterns~\cite{302672}.
Actual formulas underlying the patterns are available at~\cite{patternswebsite} and recalled in~\cite{decenttool2}.
To generate formulas, we proceeded as follows. 
For each pattern, we randomly select one of its associated formulas.
Such a formula is ``pa\-rametrized'' by some atomic
propositions.
To obtain randomly generated formula, using the distributed alphabet, we randomly instantiate atomic propositions.
\begin{table}[t]
\caption{Experiments for specification patterns}
\setlength{\extrarowheight}{1mm}
\centering
\subfloat[Number and size of messages\label{tab:bench:pattern1}]{
\begin{tabular}{|c|r|r|r|r|r|r|r|r|r|r|r|}
\hline 
$|\varphi|$ & \multicolumn{3}{c|}{$\# \mathrm{msg.}$} & \multicolumn{3}{c|}{$|\mathrm{msg.}|$} & \multicolumn{2}{c|}{$\# \mathrm{msg.}$ ratio} & \multicolumn{2}{c|}{$|\mathrm{msg.}|$ ratio} \\
\hline 
abs & 7.33 & 4.46 & 17.9 & 22 & 2,050 & 194 & 0.6 & 2.44 & 93.6 & 8.85\\
\rowcolor[gray]{0.9}
exis & 43.9 & 19.7 & 64.2 & 131 & 10,200 & 663 & 0.45 & 1.46 & 77.6 & 5.03\\
bexis & 65.3 & 31.6 & 379 & 19.6 & 1,170,000 & 5,450 & 0.48 & 2.17 & 5,970 & 10.4\\
\rowcolor[gray]{0.9}
univ & 10.3 & 5.92 & 30.9 & 31 & 2,750 & 379 & 0.57 & 2.98 & 88.6 & 12.2\\
prec & 77.6 & 25.4 & 68.1 & 232 & 8,710 & 648 & 0.32 & 1.29 & 37.4 & 4.11\\
\rowcolor[gray]{0.9}
resp & 959 & 425 & 1,070 & 2,870 & 337,000 & 9,760 & 0.44 & 1.12 & 117 & 3.39\\
precc & 7.68 & 4.81 & 18.9 & 23. & 5,180 & 218 & 0.62 & 2.47 & 225 & 9.53\\
\rowcolor[gray]{0.9}
respc & 643 & 381 & 732 & 1,920 & 719,000 & 6,680 & 0.59 & 1.13 & 372 & 3.46\\
consc & 490 & 201 & 469 & 1,470 & 337,000 & 4,260 & 0.41 & 1.13 & 229 & 3.43\\
\hline 
\end{tabular}
}

\subfloat[Trace length, delay, and memory size\label{tab:bench:pattern2}]{
\begin{tabular}{|c|r|r|r|r|r|r|r|r|r|}
\hline 
$|\varphi|$ & \multicolumn{3}{c|}{$|\mathrm{trace}|$} & \multicolumn{2}{c|}{$\mathrm{delay}$} & \multicolumn{2}{c|}{$|\mathrm{mem}|$} \\
\hline 
abs & 3.89 & 4.55 & 5.66 & 0.66 & 1.77 & 496 & 12.4\\
\rowcolor[gray]{0.9}
exis & 28.2 & 28.9 & 29.9 & 0.65 & 1.68 & 376 & 11.7\\
bexis & 42.6 & 43.1 & 116 & 0.581 & 1.56 & 28,200 & 14.4\\
\rowcolor[gray]{0.9}
univ & 5.96 & 6.73 & 7.76 & 0.76 & 1.79  & 498 & 13.0\\
prec & 50.8 & 51.6 & 35.5 & 0.81 & 1.66 & 663 & 11.5\\
\rowcolor[gray]{0.9}
resp & 638 & 639 & 639 & 0.32 & 0.7 & 1,540 & 8.61\\
precc & 4.11 & 4.82 & 5.72 & 0.7 & 1.64 & 1,200 & 11.6\\
\rowcolor[gray]{0.9}
respc & 427 & 428 & 428 & 0.59 & 1.16 & 4,650 & 10.7\\
consc & 325 & 325 & 326 & 0.6 & 1.35 & 2,720 & 10.8\\
\hline 
\end{tabular}
}
\end{table}
Results are reported in
Tables~\ref{tab:bench:pattern1} and~\ref{tab:bench:pattern2} for each kind of patterns (absence,
existence, bounded existence, universal, precedence, response, precedence chain, response chain, constrained chain), we
generated again 1,000 formulas, monitored over the same architecture as
used in Example~\ref{ex:cm1}.

%
\subsection{Conclusions from the Experiments and Discussion}
\label{sec:expe:conc}
%
The number and size of exchanged messages when monitoring with the decentralized algorithm using automata are in the same order of magnitude (and most often lower) as when monitoring with the centralized algorithm.
Comparing the decentralized monitoring algorithms, the number of messages when using LTL formulas is always lower but the size of messages is much bigger in that case (sometimes by orders of magnitude).
Delays are always greater when using automata but they remain in the same order of magnitude.
Please also note that we have conducted benchmarks where our algorithm uses only one leader monitor, which tends to augment the delay (whereas in the algorithm using LTL formulas monitors are not constrained) - see the discussion below.
Regarding the size of memory, the algorithm using automata is always more efficient by several orders of magnitude when the size of formulas grows.
\paragraph{Efficiency of Implementation.}
Another interesting feature of our algorithm is its usability in implementation.
To illustrate this point, we measured the real memory consumption of the two (reasonably optimized) implementations of benchmarks (in the same programming language), see Table~\ref{tab:perf}.
\begin{table}[t]
\caption{Evaluation of memory consumption and execution time}
\label{tab:perf}
\centering
\begin{tabular}{|c|r|r|r|r|r|r|}
\hline 
&	$\# \mathrm{msg.}$ & $|\mathrm{msg.}|$ & $|\mathrm{mem| \; (Mo)}$ & $\mathrm{time\; (s)}$ \\
\hline 
\textsc{DecentMon} & $367$ & $21,667,225$ & $157,845$ & $4.724$ \\
\rowcolor[gray]{0.9}
\textsc{DecentMon2} & $3,258$ & $59,628$ & $18$ & $0.064$ \\
\hline 
\end{tabular}
\end{table}
We only report the results when monitoring formulas of type bounded existence, over alphabet $\{a\},\{b\},\{c\}$, with a trace of 10,000 events.
For other kinds of formulas, the trend is similar.
As expected, progression is certainly more costly and thus less appropriate for monitoring.
Moreover, the size of messages (and hence the size of formulas) monitors have to handle becomes unmanageable quite rapidly.
\paragraph{Influence of the number of leaders.}
We also made some experiments (omitted for space reasons) regarding the influence of the number of leader monitors.
It turns out that, as the number of leaders augments in the system, the number of messages augments, whereas the delay induced by decentralized monitoring reduces.
For instance, by allowing all monitors to communicate spontaneously (i.e., with $\leadermon(i)= \mathit{true}$ for any $i\in [1;n]$), we observed that, for several patterns of formulas, i) a shorter average delay and less memory consumption by a factor of 1.5, and ii) the total size of messages was, in average, multiplied by 1.7 while their number was multiplied by 2 (thus the average size of messages decreased).
\section{Future Work}
\label{sec:fw}
%
Experiments in \secref{sec:expe} indicate that some parameters of our monitoring algorithm such as the frequency of communication, the number of leader monitors, and the communication architecture, influence monitoring metrics.
Our experiments allowed to sketch some empiric laws but a deeper understanding of the influence of each of these parameters is certainly needed to optimize decentralized monitoring on specific architectures.

Another line of research is related to security in decentralized monitoring, when for instance monitoring security-related properties, or when the property involves atomic propositions with confidential information.
Decentralized monitoring imposes local monitors to communicate, for instance over some network.
Exchanged messages contain information about the observation or state of monitors w.r.t. the property of interest.
Some confidentiality issues may arise.
Thus, an interesting question is to determine how and to what extent monitors could encode their local observation, transmit the encoded information, so that the message benefits (in terms of gained information) to the recipient, but not to an external observer.

We considered an architecture where communication was constrained by $\choosemon$ which can, for instance, reflect architectural constraints.
We will determine how to optimize this function according to the monitored language, the memory content, or the current state of local monitors so as to minimize e.g., exchanged messages.
%
%
\paragraph{Acknowledgment.}
The authors would like to thank Jean-Claude Fernandez and Susanne Graf for their comments on a preliminary version of this report.
%
\bibliographystyle{splncs}
\bibliography{biblio}
\newpage
\appendix
\section{Proofs}
\label{sec:proofs}
%
Let us define  a projection function $p:[1;n] \times \Sigma\rightarrow \Sigma$ s.t. $p(s,\sigma) = \sigma \cap (\bigcup_{i \in s}AP_i)$.

The following lemma says that the state obtained by applying the transition function $\delta$ of the centralized monitor always belongs to the set of states obtained by applying the transition function $\Delta_\delta$ of the corresponding decentralized monitor.
\begin{lemma} \label{contains_correct_state}
$\forall s \subseteq [1;n], \forall \sigma \in \Sigma, \forall q \in Q : \delta(q,\sigma) \in \Delta_\delta(\{q\},s,p(s,\sigma))$.
\end{lemma}
\begin{proof}
Direct, by definition of $\Delta_\delta$ (Definition~\ref{def:decentralized_mon}).
\end{proof}
The next lemma states that the function $\cQ \mapsto \Delta_\delta(\cQ,s,\sigma)$ is monotonic w.r.t. $\subseteq$.
\begin{lemma}
\label{inclusion_states}
$\forall \cQ_1 \subseteq \cQ_2 \subseteq \cQ, \forall \sigma \in \Sigma, \forall s \subseteq [1;n] : \Delta_\delta(\cQ_1,s,\sigma) \subseteq \Delta_\delta(\cQ_2,s,\sigma)$.
\end{lemma}
\begin{proof}
Consider $q_2 \in \Delta_\delta(\cQ_1,s,\sigma)$. By definition of $\Delta_\delta$, $\exists \sigma' \in \Sigma, \exists q_1 \in \cQ_1 : \sigma = p(s,\sigma') \wedge q_2 = \delta(q_1,\sigma')$.
Since $\cQ_1 \subseteq \cQ_2$, $q_1 \in \cQ_2$.
Therefore, $q_2 \in \Delta_\delta(\cQ_1,s,\sigma)$.
\end{proof}
The next lemma states that, when a monitor knows the events of all other monitors, the locally estimated global state is the actual one.
\begin{lemma}
\label{knows_everything}
If $s = [1;n]$, then $\forall \sigma \in \Sigma, \forall q \in Q : \Delta(s,\{q\},p_s(\sigma)) = \{\delta(q,\sigma)\}$.
\end{lemma}

\begin{proof}
Suppose $s = [1;n]$.
By Lemma~\ref{contains_correct_state}, $\delta(q,\sigma) \in \Delta_\delta(\{q\},s,p(s,\sigma))$.
Let us consider $q_2 \in \Delta_\delta(s,\{q\},p(s,\sigma))$.
By definition of $\Delta_\delta$, $\exists \sigma' \in \Sigma : p(s,\sigma') = \sigma \wedge q_2 = \delta(q,\sigma')$.
Since $s = [1;n]$, we have $p(s,\sigma') = \sigma' \cap (\bigcup_{1 \leq i \leq n}AP_i) = \sigma' \cap AP = \sigma'$.
Thus, $\sigma' = \sigma$.
Therefore, $q_2 = \delta(q,\sigma') = \delta(q,\sigma)$ since $\delta$ is deterministic.  
\end{proof}
The next lemma states that the function $\mathsf{update\_state}$ is well defined as the projection of an event on the union of some alphabets is equal to the union of the projections of that event on each alphabet.
\begin{lemma}
\label{correction_mem}
$\forall \sigma,\sigma',e \in \Sigma, \forall s,s' \subseteq [1;n]: \sigma = p(s,e) \wedge \sigma' = p(s',e) \implies \sigma \cup \sigma' = p(s\cup s', e)$.
\end{lemma}
\begin{proof}
We have: $p(s \cup s', e) = e \cap \left(\bigcup_{i \in s\cup s'} AP_i\right) = e \cap \left(\left(\bigcup_{i \in s} \mathit{AP}_i\right)\cup\left(\bigcup_{i \in s'}AP_i\right)\right)$.
Using the definition of the projection function $p(s \cup s', e)	= p(s,e) \cup p(s,e')=\sigma \cup \sigma' $.
\end{proof}
\begin{proof}[of Theorem 1]
Let us prove that $\delta(q,\tau(t_{\rm last},t'-1)) \in \cQ$.
Initially, we have $\delta(q,\tau(t_{\rm last}, t_{\rm last}-1)=q\in \cQ$ because $\cQ = \{q\}$.
Let us show that $\delta(q,\tau(t_{\rm last},t')) \in \cQ$ is propagated by each iteration of the for loop of the function $\mathsf{update\_state}$.
\end{proof}
\SetKwComment{tcc}{// }{}%
\begin{minipage}{0.5\linewidth}
\begin{algorithm}[H]
	$\cQ \longleftarrow \{q\}$ \;
	$\mathit{upd\_state}\longleftarrow \mathit{false}$ \;
	\For{$t'$ \textbf{from} $t_{\rm last}$ \textbf{to} $t$}{
		\tcc{1: $\delta(q,\tau(t_{\rm last},t'-1)) \in \cQ$}
		$(\sigma,s) \longleftarrow \mem(t')$ \;
		\tcc{2: $\sigma = p(s,\tau(t'))$}
		$\cQ \longleftarrow \Delta_\delta(\cQ,s,\sigma)$   \;
		\tcc{3: $\delta(q,\tau(t_{\rm last},t')) \in \cQ$}
		\If{$\exists q' \in Q: \cQ = \{q'\}$}{
			$(q,t_{\rm last}) \longleftarrow (q', t' + 1)$\;
			$\mathit{upd\_state} \longleftarrow \mathit{true}$ \;
		}		
	}
\caption{$\mathsf{update\_state}$ annotated}
\end{algorithm}
\end{minipage}
\begin{minipage}{0.48\linewidth}
We suppose that the content of the memory is correct, i.e., $\forall t'\in [t_{\rm last};t]: \mem(t') = (\sigma,s) \implies p(s,\tau(t')) = \sigma$.
We use the annotated version of the $\mathsf{update\_state}$ function and show that $(1)$ implies $(3)$.
Let us consider $q' = \delta(q,\tau(t_{\rm last},t'-1))$.
We suppose that $q' \in \cQ$ line 5.
By Lemma~\ref{inclusion_states}, $\Delta_\delta(\{q'\},s,\sigma) \subseteq \Delta_\delta(\cQ,s,\sigma)$.
By Lemma~\ref{contains_correct_state}, $\delta(q',\tau(t')) \in \Delta_\delta(s,\{q'\},p(s,\tau(t')))$.
But $p(s,\tau(t')) = \sigma$.
Thus $\delta(q',\tau(t')) = \delta(q,\tau(t_{last},t')) \in \Delta_\delta(\{q'\},s,\sigma) \subseteq \Delta_\delta(\cQ,s,\sigma)$, and $\Delta_\delta(\cQ,s,\sigma)$ is the value of $\cQ$ (line 5).
\end{minipage}

As local memories are updated with information from the component and other monitors, we deduce that the initial memory content remains correct during monitoring.

\begin{lemma}
Let $t_{\rm last}^i$ (resp. $t^i$ and $\mem_i$) be the variable $t_{\rm last}$ (resp. $t$ and memory) local to monitor $M_i$.
Let $t' \in \N$.
Consider a monitor $M_i$ such that $t^i_{\rm last} \leq t' \leq t^i$.
Let $(\sigma_i,s_i) = \mem_i(t')$.
Let $j = \choosemon(i)$.
If $M_i$ sends a message to $M_j$, then, after receiving this message, $t^j_{\rm last} > t'$ or $\mem_j(t') = (\sigma_j,s_j)$ with $s_i \subset s_j$.
\end{lemma}
\begin{proof}
Suppose that $M_i$ sends a message to $M_j$, parametrized with some time $t_{new}$. If $t_{\rm new} > t'$, the lemma holds directly. 
Otherwise, $M_i$ necessary sends its trace $\mem_i(t_{\rm new},$ $t^i)$, with possibly its state.
$M_j$ updates its trace, and tries to compute a new state.
If it finds one with $t_{\rm last}^j > t'$, then the lemma holds again.
Otherwise, after updating its memory, $\mem_j(t') = \sigma_j,s_j$ with $s_j = s_i \cup s'$, $s'$ being the previous value for the set of indexes of $mem_j(t')$.
We have $t > t'$, thus $j \in s'$.
Function $\choosemon$ is cyclic, thus $j \notin s_i$ (otherwise, we would have $s_i = [1;n]$, so by Lemma~\ref{knows_everything}, $M_i$ would have found a unique state for time $t'$, therefore sending a message parametrized with time $t_{\rm new} > t'$, which is not the case). 
Hence, $s_i \subset s_i \cup s'$.
\end{proof}
\begin{proof}[of Theorem~\ref{theo:completeness2}]
It takes at most $n-1$ communication steps to find a state corresponding to some event.
Indeed, using the previous lemma, $s$ is strictly increasing w.r.t $\subseteq$, until either the state is found, or $s=[1;n]$ and the state is also found.
\end{proof}
Regarding Theorem~\ref{theo:completeness}, principle is the same.
However, we cannot bound the number of time steps between two messages and need to suppose at least one leader component.
\end{document}